\documentclass[sigconf,nonacm]{acmart}
\usepackage{subcaption}
\usepackage{float}
\usepackage{placeins}
\usepackage{enumitem}
\usepackage{makecell}
\usepackage{tikz}
\usetikzlibrary{arrows.meta,positioning}

\newtheorem{assumption}{Assumption}
\newtheorem{proposition}{Proposition}[section]

\begin{document}

\title{Multi-Experiment Analysis}

\author{Reza Hosseini}
\affiliation{
	\institution{LinkedIn Corporation}
	\city{Sunnyvale}
	\state{CA}
	\country{USA}
}
\email{rhosseini@linkedin.com}

\begin{abstract}
Online controlled experiments face growing challenges from overlapping tests on shared traffic, where interactions between concurrent experiments obscure insights into feature combinations and produce effect estimates that do not correspond to any actionable launch scenario. While traffic splitting, layering, and sequential execution (non-concurrent) mitigate some of these issues, they require coordination overhead and can reduce experimentation velocity.
We propose Multi-Experiment Analysis (MEA), a methodology for consistent joint estimation in the presence of arbitrary partial or full overlaps and multiple variants. MEA produces three types of estimates: (1) corrected individual treatment effects that account for the presence of overlapping experiments, (2) combined effects of launching any desired combination of variants across experiments, and (3) conditional effects of an experiment's variant given that specific variants of other experiments are launched or deramped---all without requiring factorial pre-design or traffic restrictions.
We validate the approach through comprehensive simulations confirming consistency and correct coverage. We report on production deployment at scale, illustrate the methodology through real-world use cases, and share practical lessons learned---including system design, adoption patterns, and insights from production use.
\end{abstract}


\maketitle

	\section{Introduction}
	
	Online controlled experiments---also known as A/B tests,
	randomized experiments, or split tests---have become
	the gold standard for data-driven decision making in
	technology companies \cite{kohavi2007practical, kohavi2009controlled}.
	The ability to randomly assign users to treatment and control
	conditions, and measure the causal impact of product changes,
	has revolutionized how organizations develop and iterate on
	digital products \cite{kohavi2020trustworthy, thomke2020experimentation}.
	Major technology companies including Google, Microsoft, Amazon,
	Meta, Netflix, and LinkedIn now run thousands of experiments
	annually, with some organizations conducting tens of thousands
	of concurrent tests \cite{tang2010overlapping, kohavi2013kdd, larsen2024statistical}.
	These practices build on foundations dating to
	Neyman's potential outcomes
	framework~\cite{neyman1923application} and Fisher's
	randomization inference~\cite{fisher1935design}.

	Experimentation velocity---the ability to rapidly test,
	learn, and iterate on product changes---is critical
	for competitive advantage in technology, e-commerce,
	healthcare, and many other domains \cite{thomke2020experimentation}.
	Organizations running hundreds of concurrent experiments per quarter
	face persistent bottlenecks that limit their ability
	to test more ideas faster and obtain reliable results.
	
	Multiple technical approaches address experimentation
	bottlenecks.
	\textbf{Variance reduction} (VR) improves statistical
	efficiency by adjusting estimates with pre-experiment
	covariates
	\cite{freedman2008regression, lin2013agnostic,
	deng2013improving, xie2016improving,
	bloniarz2016lasso, wager2016high,
	hosseini2019unbiased, guo2021mlrate}.
	\textbf{Sequential testing} enables early stopping
	\cite{johari2017peeking}, and \textbf{surrogate
	metrics} accelerate readout with fast-moving proxies
	\cite{athey2019surrogate}.
	\textbf{Factorial designs}
	\cite{fisher1935design, wu2009experiments,
	box2005statistics} enable joint estimation of
	multiple factors but require pre-planned
	coordination, complete population overlap, and
	balanced allocation---impractical when experiments
	are launched independently by different teams with
	different triggering conditions and timelines.

	When multiple tests run concurrently on shared user
	populations, two problems arise:
	\emph{collisions}, where simultaneous experiments
	break user experience and must be prevented via
	traffic splitting or sequential execution; and
	\emph{interactions}, where experiments run together
	without breaking UX but bias naive estimates of
	individual effects and hide optimal combinations.
	Major platforms address overlaps through
	traffic layering and isolation
	\cite{tang2010overlapping, kohavi2013kdd,
	xu2015infrastructure, gupta2018anatomy}, which
	partition traffic so that experiments within the
	same layer do not interact.
	Recent empirical work further suggests that most
	experiment interactions are small
	\cite{jeng2023relax, chan2023embrace, larsen2024statistical},
	and screening all pairwise interactions may not be
	warranted.

	We introduce Multi-Experiment Analysis (MEA), a
	post-hoc framework for consistent analysis of
	overlapping experiments that are safe to run
	concurrently, without requiring traffic restrictions
	or sequential execution.
	MEA recovers insights analogous to factorial designs
	from naturally occurring partial or full overlaps,
	handling the reality of independent experiments
	with different triggering conditions and timelines.
	MEA is designed for situations where
	interactions do occur or cannot be ruled out:
	diagnosing suspected interactions after the fact
	when unexpected metric movements arise;
	quantifying the impact of uncontrollable concurrent
	changes such as company-wide migrations or
	platform updates;
	enabling parallel experimentation when long-readout
	metrics (e.g., retention, lifetime value) make
	temporal sequencing infeasible;
	analyzing same-surface or same-funnel overlaps
	where triggering populations obviously intersect;
	and estimating one experiment's effect conditional
	on another team's treatment being deployed---a
	scenario outside one's control.
	Critically, MEA does not merely test whether
	interactions exist---it estimates the impact of
	whatever launch scenario is chosen (launch both,
	launch only one, or estimate the effect assuming
	the other treatment ships), with the interaction
	diagnostic as a secondary output.

	While the underlying statistical
	techniques---post-stratification and weighted
	averaging within strata---are well established
	\cite{miratrix2013post, lohr2021sampling}, to our
	knowledge this work is the first to systematically
	apply them to causal estimation across overlapping
	online experiments with complex co-triggering
	patterns.
	The approach is entirely non-parametric---relying on
	weighted averages within strata rather than
	regression models or distributional
	assumptions---making it robust and broadly
	applicable.
	We formalize the causal assumptions
	under which MEA yields consistent estimates---in particular, \emph{Arm-Trigger Invariance},
	which requires that treatment assignments do not
	shift triggering probabilities across experiments
	--- and develop automated diagnostics that test
	this condition using experiment data alone.
	We also provide a systematic framework for
	conditional (scenario-based) estimation, addressing
	a common practitioner need: understanding the effect
	of one experiment given anticipated launch decisions
	for others.
	Beyond methodology, we report on production
	deployment at scale and share practical lessons
	learned.

	\section{Methodology}

	MEA applies to overlapping experiments that share
	orthogonal randomization---each experiment assigns
	units independently via hashing of user
	identifiers \cite{kohavi2020trustworthy}---and
	have sufficient temporal overlap to observe
	relevant trigger combinations.
	It is important to distinguish MEA from the
	literature on interference, where one unit's
	treatment affects another unit's outcome
	\cite{aronow2017estimating}.
	MEA addresses a different problem: bias from
	analyzing overlapping experiments separately,
	not cross-unit spillovers.
	MEA assumes SUTVA
	\cite{imbens2015causal}---each unit's outcome
	depends only on its own treatment assignments,
	not on others'.

	The framework applies to any number of experiments
	and arbitrary variants per experiment, fulfilling
	two core objectives: (1) identifying the optimal
	combination of features, and (2) estimating the
	conditional impact of one experiment given specific
	launch decisions for others.
	We illustrate with two experiments ($E_1$, $E_2$),
	but the framework generalizes to arbitrary numbers
	of experiments, variants, and overlap patterns.

	\subsection{Triggering-Based Population Partitioning}
	
	When experiments overlap, units are partitioned
	based on joint triggering logic.
	Let $A_i$ denote the arm assignment in
	experiment~$i$ (a latent label set before
	triggering),
	$S_i = \mathbf{1}\{\text{user triggers }E_i\}$
	the trigger indicator, and
	$V_i = A_i$ when $S_i = 1$,
	$V_i = \texttt{nan}$ otherwise (the observed
	variant).
	The trigger-state vector
	$S = (S_1, \ldots, S_k)$ defines mutually exclusive
	regions (e.g., four for two experiments):
	\begin{itemize}
		\item $\mathcal{R}_{10}$: only $E_1$ triggered ($S_1{=}1, S_2{=}0$)
		\item $\mathcal{R}_{01}$: only $E_2$ triggered ($S_1{=}0, S_2{=}1$)
		\item $\mathcal{R}_{11}$: both triggered ($S_1{=}1, S_2{=}1$)
		\item $\mathcal{R}_{00}$: neither triggered ($S_1{=}0, S_2{=}0$)
	\end{itemize}
	
	This partitioning isolates comparable
	sub-populations, analogous to post-stratification
	in causal inference or survey sampling, where
	conditioning on strata (here, trigger combinations)
	removes bias from confounding factors
	\cite{miratrix2013post, neyman1923application}.
	
	Figure~\ref{fig:lshape} illustrates the partitioning structure defined by 
	\emph{triggering regions}---subsets of the population determined 
	solely by which experiments were triggered, independent of variant assignment.
	These regions ($\mathcal{R}_{10}$, $\mathcal{R}_{01}$, $\mathcal{R}_{11}$) 
	form the basis for weighted estimation, with weights proportional to their respective sizes. 
	Collectively, these regions constitute the \emph{impacted surface} of the two experiments,
	whereas $\mathcal{R}_{00}$ represents the unimpacted partition.

	A key observation---often overlooked---is that
	the impact of launching $(t_1, t_2)$ is not
	restricted to $\mathcal{R}_{11}$: in
	$\mathcal{R}_{10}$, users experience only
	$t_1$ vs.\ $c_1$; in $\mathcal{R}_{01}$,
	only $t_2$ vs.\ $c_2$.
	The combination launch impact is a weighted
	average across \emph{all} triggered regions,
	not just the intersection.

	Figure~\ref{fig:variants} shows the finer partition
	within each triggering region: \emph{variant cells}
	defined by specific variant combinations. Each tuple
	denotes a variant combination and ``nan'' represents
	lack of triggering. For example, $(t_1, c_2)$ represents
	units where Experiment 1 treatment ($t_1$) triggered
	and Experiment 2 control ($c_2$) triggered---a variant
	cell within $\mathcal{R}_{11}$. Similarly,
	$(t_1, \texttt{nan})$ denotes a variant cell within
	$\mathcal{R}_{10}$ where only Experiment 1 triggered
	with $t_1$ and Experiment 2 did not trigger.
	
	\begin{figure}[t]
		\centering
		\includegraphics[width=0.95\columnwidth]{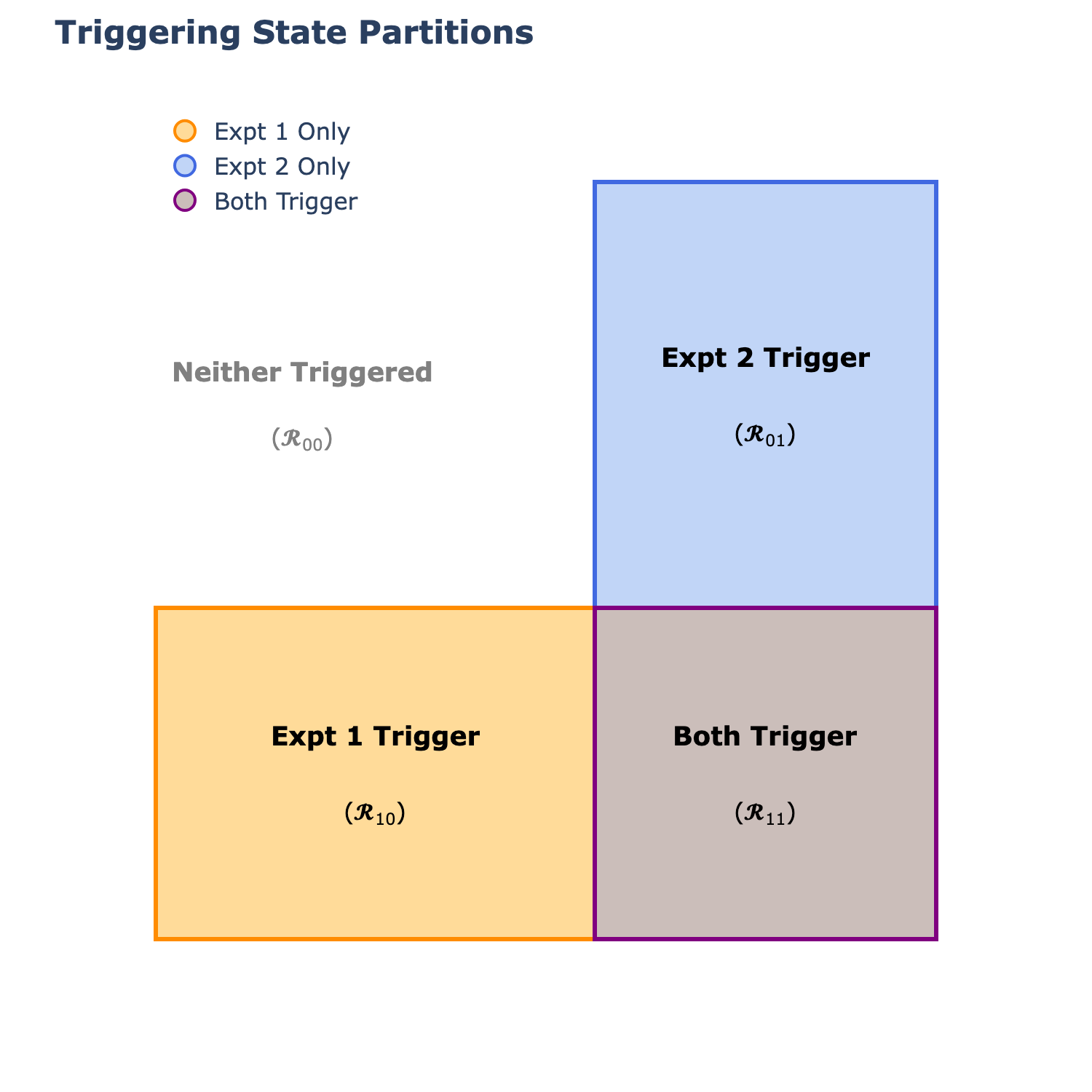}
		\caption{L-shape partitioning showing triggering regions
			$\mathcal{R}_{10}$ (orange), $\mathcal{R}_{01}$ (blue), and
			$\mathcal{R}_{11}$ (brown) for two overlapping experiments.
			These regions depend only on which experiments triggered,
			not on variant assignment.}
		\label{fig:lshape}
	\end{figure}
	
	\begin{figure}[t]
		\centering
		\includegraphics[width=0.95\columnwidth]{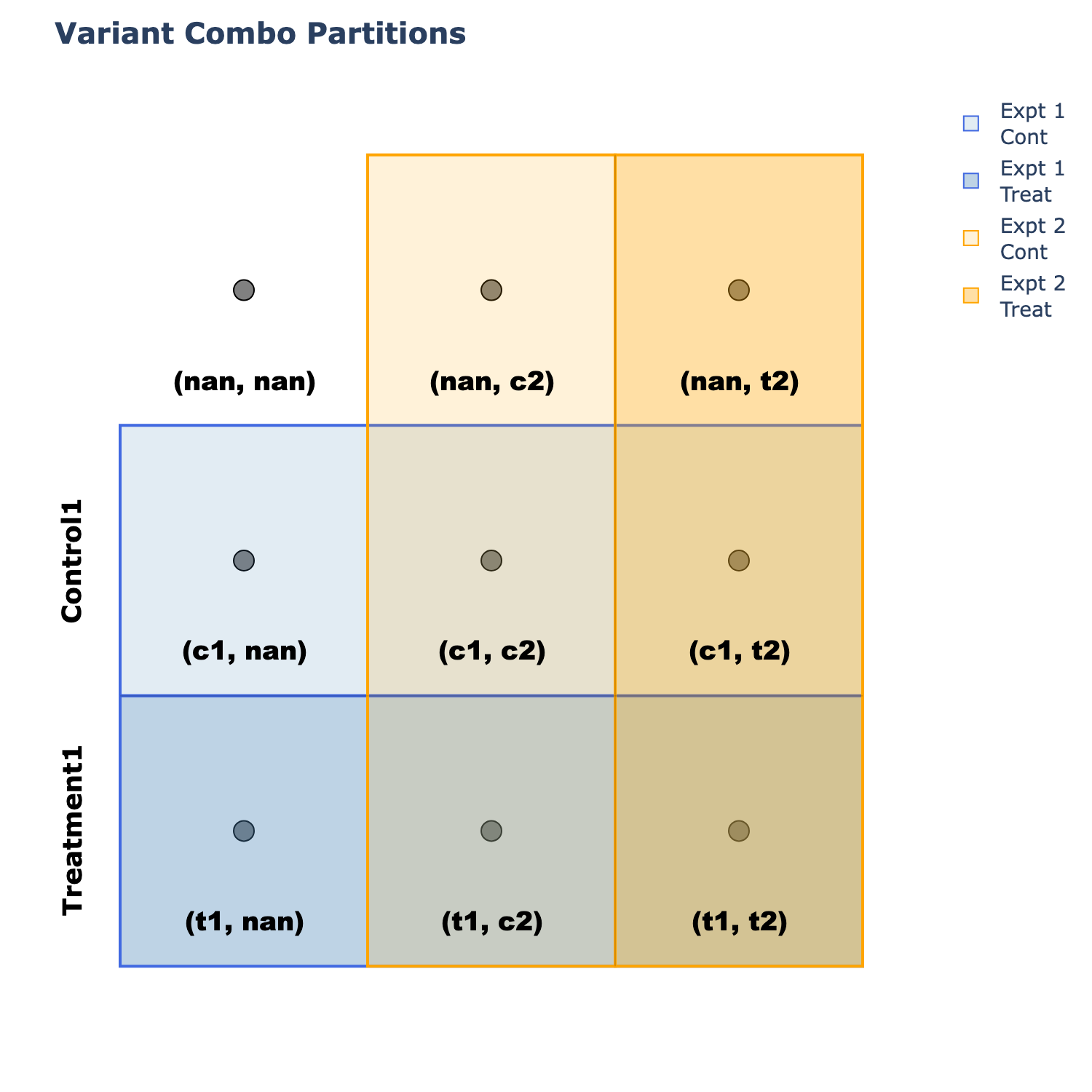}
		\caption{Variant cells within each triggering region,
			showing all variant combinations. Estimation compares
			relevant cells within each region.}
		\label{fig:variants}
	\end{figure}

	\subsection{Assumptions}

	The estimation approach builds on
	post-stratification \cite{miratrix2013post}, a
	classical technique in survey sampling
	\cite{lohr2021sampling} and causal inference
	\cite{hernan2020causal}.
	For this to produce valid causal estimates,
	two conditions must hold:
	\begin{enumerate}[nosep]
	\item \emph{Stable stratum weights
	  (Arm-Trigger Invariance).}
	  The triggering-region proportions must not shift
	  when a particular arm combination is launched:
	  $P(S{=}s \mid \mathrm{do}(A{=}a)) = P(S{=}s)$
	  for all $a$~\cite{pearl2009causality};
	  otherwise the weighted estimate targets the wrong
	  population mixture.
	  In practice, this means no experiment's treatment
	  assignment changes whether users trigger into
	  another experiment.
	\item \emph{Unbiased within-stratum estimation.}
	  Arm assignments must be independent within each
	  stratum:
	  $A_i \perp\!\!\!\perp A_j \mid S{=}s$
	  for all $i \neq j$, so that cell means are
	  unconfounded.
	\end{enumerate}
	Condition~(1) is the non-trivial causal
	requirement; condition~(2) follows from
	condition~(1) combined with independent hashing
	\cite{tang2010overlapping, kohavi2020trustworthy}.

	\textbf{Assumption not needed:} independence
	of trigger states. $S_i$ and $S_j$ may be
	freely correlated---and typically are, through
	shared user characteristics or even direct
	causal links (e.g., triggering $E_1$ makes
	triggering $E_2$ more likely)---the assumption
	only requires that \emph{arm assignments} do
	not shift trigger probabilities.
	Arm-Trigger Invariance can be violated when one
	experiment's treatment changes user behavior in a
	way that affects triggering of another experiment---for example, if $E_1$'s treatment adds a
	prominent link to a page where $E_2$ triggers,
	users in $E_1$'s treatment arm will trigger $E_2$
	at a higher rate than control users, biasing the
	overlap population.
	
	Appendix~\ref{sec:assumption-checking} formalizes
	Arm-Trigger Invariance via a causal graph,
	proves the proposition above, and develops
	automated chi-squared diagnostics (using
	Cram\'{e}r's~$V$ as an effect-size gate)
	that test this condition using experiment data
	alone, without requiring external knowledge of
	experiment designs.
	Every production MEA report includes both
	the statistical test and a grouped bar-chart
	visualization (see
	Figure~\ref{fig:assumption-check} for a
	production case where the assumption is
	violated) so that users can inspect the
	conditional variant distributions directly.

	\begin{figure}[t]
	\centering
	\includegraphics[width=\columnwidth]{%
	  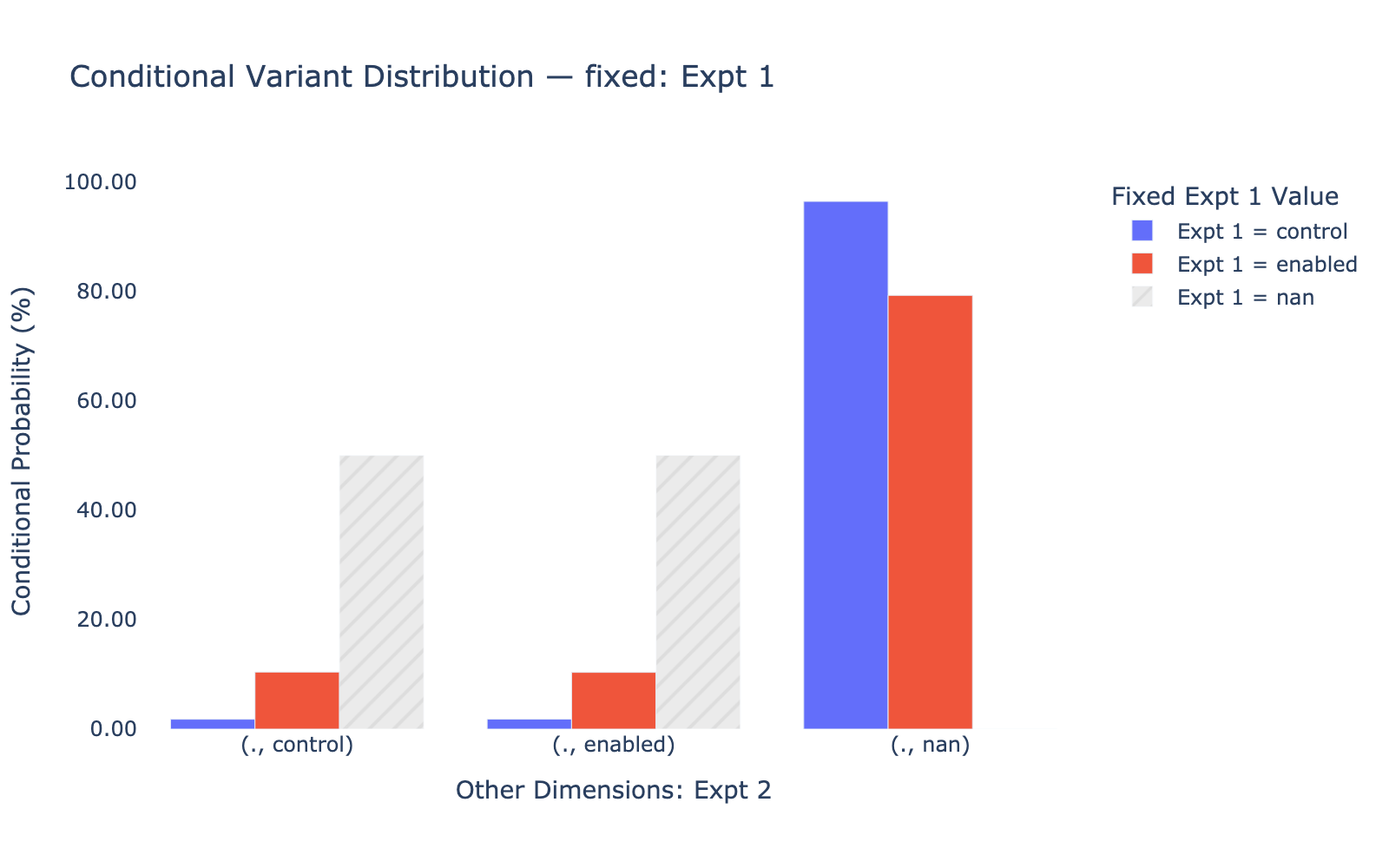}
	\caption{Production violation (FLAG):
	  two notification experiments.
	  Fixing $E_1$, the \texttt{nan} column
	  shows 96\% of control members never trigger
	  $E_2$ vs.\ 79\% of treatment members ---
	  treatment in $E_1$ dramatically increases
	  $E_2$'s trigger rate ($V = 0.262$).
	  Non-\texttt{nan} columns (control, enabled)
	  are balanced, confirming assignment
	  independence within $\mathcal{R}_{11}$.
	  The \texttt{nan} bars are greyed out because
	  trigger-state correlation is \emph{not} a
	  required assumption: only the
	  non-\texttt{nan} columns need to match across
	  arms for Arm-Trigger Invariance to hold.}
	\label{fig:assumption-check}
	\end{figure}

	\subsection{Consistent Weighted Estimation}

	Estimation proceeds in two stages: (1) compute weights
	at the triggering region level based on region sizes,
	and (2) within each triggering region, compare the
	relevant variant cells.
	
	The effect for a treatment combination is estimated
	as a weighted average of region-specific differences:
	\begin{equation}
		\hat{\tau} = \sum_r w_r \, (\bar{Y}_{t,r} - \bar{Y}_{c,r})
	\end{equation}
	where $\bar{Y}_{t,r}$ and $\bar{Y}_{c,r}$ are mean
	outcomes for the target and baseline variant cells
	within triggering region $r$, and $w_r = n_r / \sum_{r'} n_{r'}$
	are proportional to triggering region sizes.
	
	\textbf{Impacted population weighting.} In practice,
	we estimate effects over the \emph{impacted population}---units
	who triggered at least one experiment. Units in
	$\mathcal{R}_{00}$ (neither experiment triggered)
	experience no treatment effect regardless of the
	launch decision and are excluded to increase
	statistical power, as any observed differences
	in $\mathcal{R}_{00}$ are due solely to randomization noise:
	\begin{equation}
		w_r = \frac{n_r}{\sum_{r' \neq \mathcal{R}_{00}} n_{r'}}
		\quad \text{for } r \neq \mathcal{R}_{00}
	\end{equation}
	This ensures the estimated effect reflects the impact
	on the population actually affected by the launch decision.
	For site-wide metric estimation, the impacted-population
	effect can be appropriately diluted by the fraction of
	triggered users to recover the full-population effect;
	however, focusing estimation on the impacted population
	yields tighter confidence intervals without sacrificing
	validity for launch decisions.
	
	This estimator is consistent by the same
	logic as post-stratification, relying on both
	assumptions from Section~2.2:
	stable stratum weights ensure the weights
	$w_r$ target the correct population mixture;
	within-stratum independence ensures that
	cell means are unbiased.
	Together, reweighting by region proportions recovers
	the population-level causal effect
	\cite{miratrix2013post,
		lohr2021sampling, li2017general, imbens2015causal}.
	
	\textbf{Variance estimation.} Since variant assignments
	are independent across triggering regions under
	orthogonal randomization, the variance is:
	\begin{equation}
		\widehat{\mathrm{Var}}(\hat{\tau}) = \sum_r w_r^2 \, 
		\widehat{\mathrm{Var}}(\bar{Y}_{t,r} - \bar{Y}_{c,r})
	\end{equation}
	
	\textbf{Weight uncertainty.} The formula above treats
	weights as fixed. When estimated from data, weights
	are random variables. Because triggering determines
	partition membership while effect estimates depend
	on variant assignments within partitions, and these
	are independent under orthogonal randomization,
	the variance decomposes as:
	\begin{equation}
		\mathrm{Var}(\hat{w}_r \hat{\delta}_r) = 
		\mathbb{E}[\hat{w}_r]^2 \mathrm{Var}(\hat{\delta}_r) + 
		\mathbb{E}[\hat{\delta}_r]^2 \mathrm{Var}(\hat{w}_r) + 
		\mathrm{Var}(\hat{w}_r) \mathrm{Var}(\hat{\delta}_r)
	\end{equation}
	The standard formula captures only the first term.
	For large samples the additional terms are negligible,
	but for smaller samples, ignoring weight uncertainty
	can underestimate standard errors. Bucketed jackknife
	(described below) automatically captures all sources
	of variability through resampling.
	
	The approach is validated through simulations
	(Appendix~\ref{sec:simulations}) confirming
	consistency and correct coverage.
	
	\subsection{Algorithm 1: Combination Impact Analysis}
	
	To estimate the effect of launching combination
	$(t_1, t_2)$ versus baseline $(c_1, c_2)$:
	
	\begin{enumerate}
		\item \textbf{Identify variant cells to compare
			within each triggering region:}
		\begin{itemize}
			\item $\mathcal{R}_{11}$: $(t_1, t_2)$ vs $(c_1, c_2)$
			\item $\mathcal{R}_{10}$: $(t_1, \texttt{nan})$ vs $(c_1, \texttt{nan})$
			\item $\mathcal{R}_{01}$: $(\texttt{nan}, t_2)$ vs $(\texttt{nan}, c_2)$
		\end{itemize}
		
		\item \textbf{Compute region-specific effects:}
		$\hat{\delta}_r = \bar{Y}_{t,r} - \bar{Y}_{c,r}$
		
		\item \textbf{Aggregate:}
		$\hat{\tau}_{(t_1, t_2)} = \sum_r w_r \, \hat{\delta}_r$
		
		\item \textbf{Select optimal:} Repeat for all
		combinations and identify the one maximizing
		the objective metric.
	\end{enumerate}
	
	Figure~\ref{fig:algo1} illustrates which variant
	cells contribute (checkmarks) and which are
	irrelevant (×) for estimating $(t_1, t_2)$ vs $(c_1, c_2)$.
	Interactions can be positive (synergistic) or negative
	(antagonistic); MEA quantifies both.
	
	MEA reports a $p$-value for each combination effect
	estimate alongside a Bonferroni-adjusted significance
	threshold $\alpha / (\text{number of combinations})$,
	so users can directly assess which effects are
	statistically significant after multiplicity correction.
	Note that corrected individual effects are a special
	case: the effect of a single experiment's variant
	(e.g., $(t_1, c_2)$ vs.\ $(c_1, c_2)$) is estimated
	by the same weighted procedure, automatically
	accounting for the concurrent state of other
	experiments.

	\subsection{Algorithm 2: Scenario-Based Analysis}
	
	To estimate the conditional effect of $E_2 = t_2$
	given $E_1 = t_1$ is launched:
	
	\begin{enumerate}
		\item \textbf{Fix baseline scenario:} Set $E_1 = t_1$
		as the assumed launch state.
		
		\item \textbf{Identify relevant comparisons:}
		\begin{itemize}
			\item $\mathcal{R}_{11}$: $(t_1, t_2)$ vs $(t_1, c_2)$
			\item $\mathcal{R}_{01}$: $(\texttt{nan}, t_2)$ vs $(\texttt{nan}, c_2)$
			\item $\mathcal{R}_{10}$: Excluded---$E_2$ did not
			trigger, so outcomes are unaffected by $E_2$ launch
		\end{itemize}
		
		\item \textbf{Recompute weights:} Over relevant
		regions only ($\mathcal{R}_{11}$ and $\mathcal{R}_{01}$)
		
		\item \textbf{Aggregate:} Weighted sum of
		region-specific effects
	\end{enumerate}
	
	Figure~\ref{fig:algo2} shows relevant combinations
	(checkmarks) for conditional analysis and excluded
	regions (×).
	
	\begin{figure}[t]
		\centering
		\includegraphics[width=0.95\columnwidth]{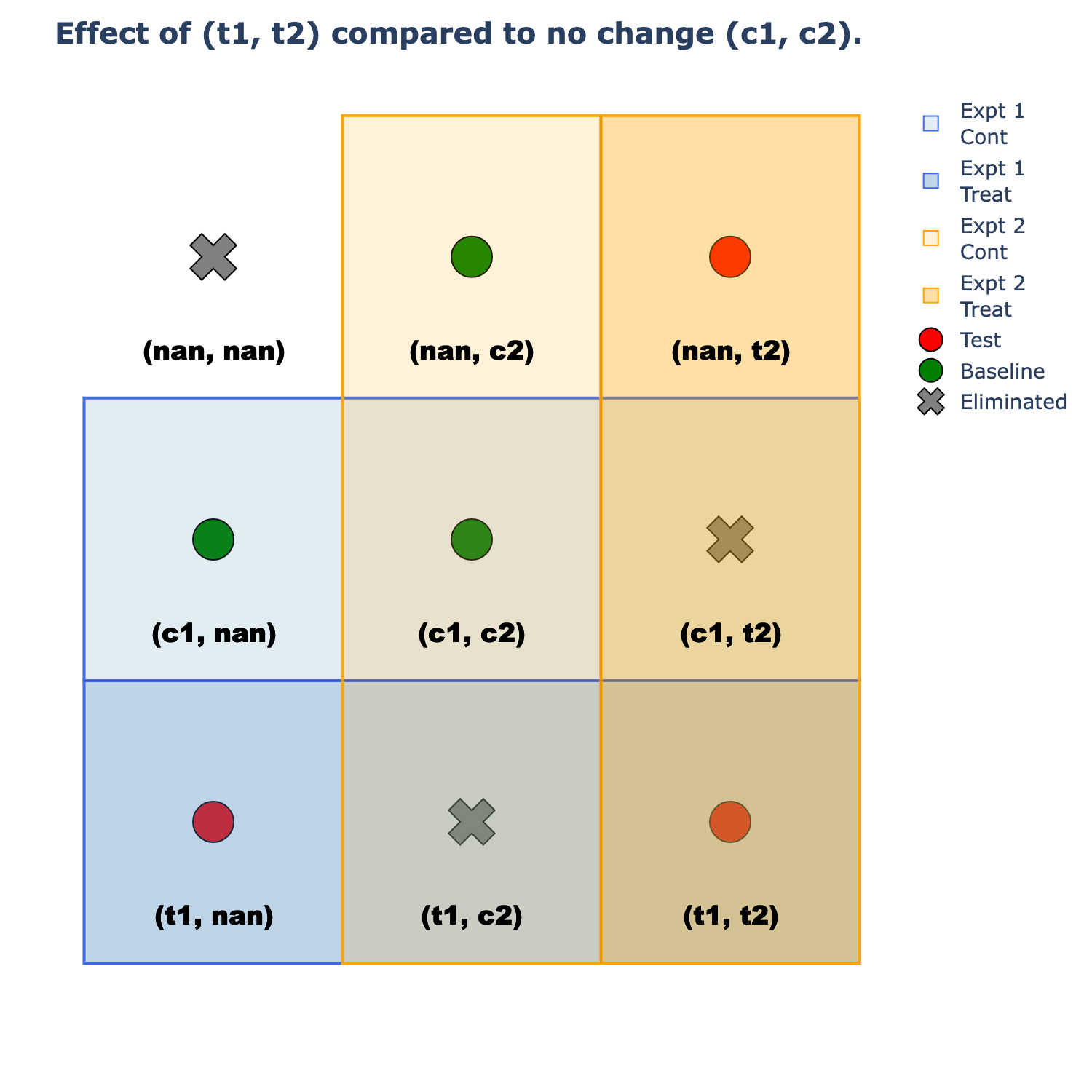}
		\caption{Combination analysis: checkmarks indicate
			relevant variant cells; × marks irrelevant ones.}
		\label{fig:algo1}
	\end{figure}
	
	\begin{figure}[t]
		\centering
		\includegraphics[width=0.95\columnwidth]{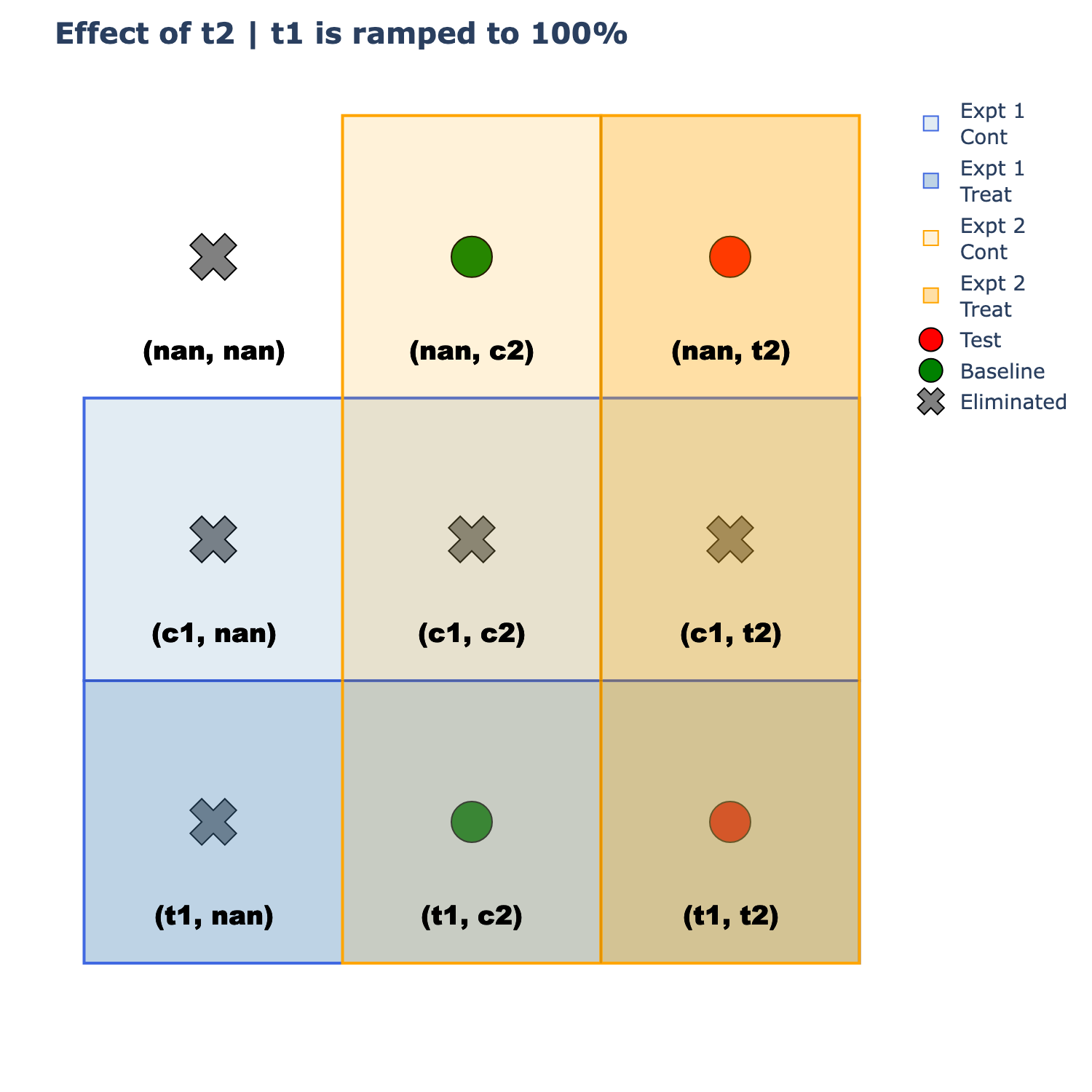}
		\caption{Conditional analysis: only cells matching
			the fixed scenario (checkmarks) are used.}
		\label{fig:algo2}
	\end{figure}
	
	\subsection{Extension to Ratio Metrics}
	\label{sec:ratio-metrics}

	Many important online experiment metrics are ratios,
	such as conversion rates, click-through rates, and
	retention metrics. MEA naturally extends to ratio
	metrics through the same partitioning framework.
	
	For a ratio metric $\theta = \mu_N / \mu_D$ (where
	$\mu_N$ is the numerator mean and $\mu_D$ is the
	denominator mean), we apply the weighted estimation
	separately to numerator and denominator:
	
	\begin{equation}
		\hat{\mu}_{N,t} = \sum_r w_r \bar{N}_{t,r}, \quad
		\hat{\mu}_{N,c} = \sum_r w_r \bar{N}_{c,r}
	\end{equation}
	\begin{equation}
		\hat{\mu}_{D,t} = \sum_r w_r \bar{D}_{t,r}, \quad
		\hat{\mu}_{D,c} = \sum_r w_r \bar{D}_{c,r}
	\end{equation}
	
	The treatment effect estimator for the ratio metric
	is then:
	\begin{equation}
		\hat{\tau}_{\theta} = \frac{\hat{\mu}_{N,t}}{\hat{\mu}_{D,t}} - \frac{\hat{\mu}_{N,c}}{\hat{\mu}_{D,c}}
	\end{equation}
	
	By the Delta method \cite{van2000asymptotic}, this
	estimator is consistent and asymptotically normally
	distributed under standard regularity conditions.
	This approach aligns with how ratio metrics
	are handled in standard A/B testing practice
	\cite{deng2013improving, kohavi2020trustworthy}.
	
	For variance estimation, while the Delta method
	provides analytic formulas, we employ \textbf{bucketed
		jackknife} \cite{efron1994introduction} in practice
	due to its broader applicability and more tractable
	implementation. The bucketed jackknife partitions
	data into $B$ buckets (typically $B = 20$--$50$)
	and computes leave-one-bucket-out estimates, requiring
	only $B$ recomputations rather than leave-one-out
	estimation for each observation. This makes it
	computationally feasible even with hundreds of
	millions of units.
	
	The bucketed jackknife variance estimator is:
	\begin{equation}
		\widehat{\mathrm{Var}}_{BJ}(\hat{\tau}_{\theta}) = 
		\frac{B-1}{B} \sum_{b=1}^{B} (\hat{\tau}_{\theta}^{(-b)} - \bar{\hat{\tau}}_{\theta})^2
	\end{equation}
	where $\hat{\tau}_{\theta}^{(-b)}$ is the estimate
	with bucket $b$ removed and $\bar{\hat{\tau}}_{\theta}$
	is the mean of the jackknife estimates.
	
	Crucially, bucketed jackknife applies uniformly to
	any metric type (means, ratios, quantiles, complex
	transformations) without requiring metric-specific
	variance derivations. This generality simplifies
	implementation and enables MEA to support arbitrary
	user-defined metrics with minimal code changes.
	Additionally, because the jackknife resamples both
	the outcome data and the region membership indicators,
	it automatically captures uncertainty from weight
	estimation as discussed above.

	\subsection{Multiple Testing and Power}

	The number of variant combinations grows
	multiplicatively: with $k$ experiments
	having $\ell_1, \ell_2, \ldots, \ell_k$ variants
	respectively, there are
	$\prod_{i=1}^{k} \ell_i - 1$ non-baseline
	combinations.
	This combinatorial growth creates two
	challenges---inflated family-wise error rates and
	reduced per-combination power---that affect any
	joint analysis, whether MEA or a coordinated
	factorial design, equally: both must test the same
	set of combinations and apply the same multiplicity
	correction (e.g., Bonferroni at
	$\alpha / (\prod \ell_i - 1)$)
	\cite{benjamini1995controlling}.
	As a practical guideline,
	two experiments with three variants each yield
	$3 \times 3 - 1 = 8$ combinations, while three
	experiments with two variants each yield
	$2 \times 2 \times 2 - 1 = 7$.
	We recommend focusing on two- or
	three-way overlaps and limiting analysis to key
	variant subsets rather than exhaustive enumeration.

	Where MEA differs from a factorial is in
	\emph{statistical power}.
	Under mild assumptions (balanced randomization,
	homoscedasticity),
	Appendix~\ref{sec:mea-sample-size} shows that
	MEA's variance for any combination effect is at
	most that of a coordinated factorial, with
	exponentially growing savings as $k$ increases:
	for $k$ binary experiments with 50\% trigger rates,
	the variance ratio scales as $(3/4)^k$.
	The efficiency gain comes from two sources:
	singly-triggered strata involve fewer variant cells
	than the full factorial grid, and stratification by
	triggering pattern reduces within-cell variance.


\section{Example Use Case}

We present a real-world application of MEA
to illustrate its practical value in production
experimentation.

\textit{Note: Effect sizes are reported as relative
	differences, the metric is monotonically transformed,
	and exact sample sizes are omitted (the analysis
	involved several million units).}

\subsection{Motivation and Setup}

A product team sought to optimize conversion rates
on a high-traffic product page that serves as an
entry point into a subscription conversion funnel.
Two independent design changes were identified:

\begin{itemize}
	\item \textbf{Experiment 1 (Content Test):}
	Five variants testing different content and messaging
	on a conversion element embedded in the page
	(control, v1, v2, v3, v4)
	\item \textbf{Experiment 2 (Page Redesign):}
	Three variants testing a full redesign of the
	page layout (control, v1, v2)
\end{itemize}

A full factorial design was not feasible---not only
due to engineering complexity, but because the two
features had different triggering conditions: the
conversion element (Expt 1) triggered only for
members who reached that element within the page,
while the page redesign (Expt 2) triggered broadly
for all page visitors. This naturally resulted in
an asymmetric overlap pattern where Experiment 1
was largely (but not fully) nested within
Experiment 2's triggering population.

The primary metric of interest was conversion rate.
Traditional univariate analysis of each experiment
in isolation would not reveal the optimal combination
or detect potential interactions.

\subsection{MEA Application and Results}

Table~\ref{tab:usecase_overlap} shows the overlap
diagnostics from MEA, revealing the asymmetric
triggering pattern: 86.5\% of Expt 1 units also
triggered Expt 2, while only 64.4\% of Expt 2 units
triggered Expt 1.

\begin{table}[t]
	\centering
	\caption{Overlap diagnostics showing asymmetric
		triggering overlap between experiments.}
	\label{tab:usecase_overlap}
	\small
	\begin{tabular}{lr}
		\toprule
		Trigger State & Percent \\
		\midrule
		Both experiments & 58.5\% \\
		Only Expt 1 & 9.1\% \\
		Only Expt 2 & 32.3\% \\
		\bottomrule
	\end{tabular}
\end{table}

\begin{table}[t]
	\centering
	\caption{Univariate analysis results (each experiment analyzed in isolation).}
	\label{tab:usecase_univariate}
	\small
	\begin{tabular}{llr}
		\toprule
		Experiment & Best Variant & Delta (\%) \\
		\midrule
		Expt 1 (Content) & v4 & +36.9 \\
		Expt 2 (Redesign) & v2 & +26.2 \\
		\bottomrule
	\end{tabular}
\end{table}

Table~\ref{tab:usecase_univariate} shows the univariate
results for each experiment analyzed in isolation.
The univariate effect of Expt 2 v2 (+26.2\%) averages
over all concurrent states of Expt 1---a mixture that
is an artifact of the experimental period and will not
exist post-launch, when Expt 1 will be in a specific
decided state. MEA decomposes this into scenario-specific
estimates; for example, the effect of Expt 2 v2 when
Expt 1 is held at control is +48.6\%, as shown in
Table~\ref{tab:usecase_bias}.

\begin{table}[t]
	\centering
	\caption{Univariate vs MEA scenario-specific estimate for Expt 2 v2.}
	\label{tab:usecase_bias}
	\small
	\begin{tabular}{lr}
		\toprule
		Analysis & Expt 2 v2 Effect (\%) \\
		\midrule
		Univariate & +26.2 \\
		MEA (control, v2) & +48.6 \\
		\bottomrule
	\end{tabular}
\end{table}

MEA identifies (v4, v2) as the optimal combination,
as shown in Table~\ref{tab:usecase_results}.

\begin{table}[t]
	\centering
	\caption{MEA combination effects for top variants.}
	\label{tab:usecase_results}
	\small
	\begin{tabular}{lr}
		\toprule
		Combination & Delta (\%) \\
		\midrule
		(v4, v2) & +96.0 \\
		(v4, control) & +62.3 \\
		(control, v2) & +48.6 \\
		\bottomrule
	\end{tabular}
\end{table}

\subsection{Decision Impact}

The MEA analysis revealed several key findings:

\begin{enumerate}
	\item \textbf{Estimand mismatch}: The univariate estimate
	of Expt 2 v2 (+26.2\%) does not correspond to any
	actionable launch scenario. MEA reveals that the
	scenario-specific effect ranges from +48.6\% (Expt 1
	at control) to other values depending on which Expt 1
	variant is launched---a difference of over 22 percentage
	points from the univariate number.
	\item \textbf{Non-additive interaction}: Individual effects
	do not sum to the combined effect, whether measured as
	absolute mean differences or relative percent changes
	(we present only percent changes here for confidentiality).
	This demonstrates the presence of meaningful interactions
	between the two experiments, which univariate analysis
	would entirely miss.
	\item \textbf{Optimal combination}: MEA identifies
	(v4, v2) as the best-performing combination with
	+96\% lift, a conclusion that requires joint analysis
	and cannot be reached by analyzing each experiment
	in isolation.
\end{enumerate}

Based on MEA results, the team confidently launched
the (v4, v2) combination. This decision would not have
been possible with sequential execution (which would have
delayed launch by weeks) or univariate analysis (which
would have missed both the estimand mismatch and the
interaction effect).

The ability to test both changes simultaneously while
recovering consistent combination estimates reduced
time-to-market and enabled data-driven selection of
the optimal feature set.
Appendix~\ref{sec:concurrent-vs-sequential}
formalizes this advantage through a simulation showing
that univariate analysis and sequential execution both
recommend the worst launch combination, while MEA
correctly identifies the optimum.

	\section{System Design and Implementation}
	
	MEA is implemented as a self-serve, on-demand tool.
	As noted in Section~1, most experiment
	interactions are small
	\cite{jeng2023relax, chan2023embrace, larsen2024statistical},
	and MEA is not needed for every pair.
	Pre-computing all possible overlaps is computationally
	infeasible ($O(n(n-1)/2)$ for pairs, $2^n$ for subsets).
	On-demand processing concentrates resources on
	user-specified high-priority overlap sets.
	
	The system is designed to handle experiments with
	hundreds of millions of units and hundreds of metrics.
	Figure~\ref{fig:system} illustrates the overall
	system architecture.
	
	\subsection{Query Construction and Execution Flow}
	
	The core computation is designed around efficient,
	scalable data processing, as shown in
	Figure~\ref{fig:system}:
	
	\begin{itemize}
		\item Assignment queries are dynamically constructed
		based on user-provided experiment IDs, variants,
		and triggering conditions (or standard exit logs).
		\item These assignment queries are joined with metric
		tables (pre-registered or user-specified) to produce
		a unified table containing unit-level assignment
		and metric values.
		\item The heavy join and filtering operations are
		executed via distributed query engines.
		Two back-ends are supported:
		Trino for interactive, smaller-scale analyses
		and Spark for larger jobs or scheduled
		recurring runs.
		\item Once materialized, summary statistics (means,
		counts, variances) are computed per
		combination for each metric---again using the
		distributed engine for parallelism across hundreds
		of metrics.
		\item The final small aggregated table (
		combination × metric summaries) is brought into
		memory for the lightweight weighted estimation
		and variance computation.
	\end{itemize}
	
	This hybrid approach---heavy lifting in distributed SQL,
	lightweight final math in memory---enables both massive
	scale and fast end-to-end runtime.
	The overall complexity is $O(N)$:
	MEA's $k$-way join is a sequence of $k$ hash joins
	(each $O(N)$, identical to joining $k$ independent
	experiments), aggregation adds one \texttt{GROUP BY}
	dimension for jackknife buckets, and variance
	estimation operates on the small summary table
	with no additional data scan.
	In production, MEA runs in minutes on the Trino
	back-end for experiments with tens of millions of
	members per arm, and in under 30 minutes on the
	Spark back-end for datasets with hundreds of
	millions of members.
	
	\begin{figure}[t]
		\centering
		\includegraphics[width=0.95\columnwidth]{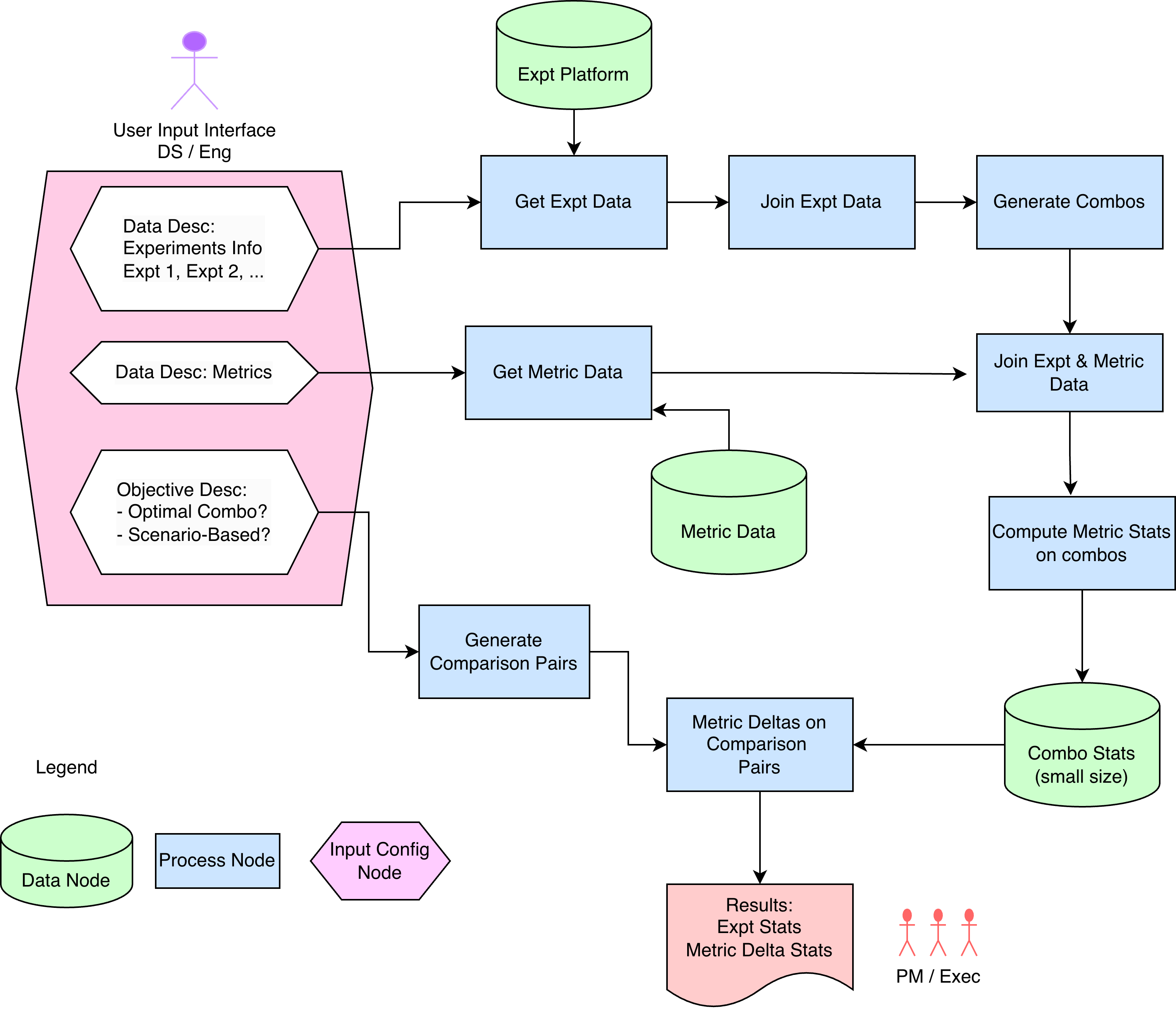}
		\caption{MEA system architecture showing the hybrid
			approach combining distributed computation for
			heavy data processing with in-memory computation
			for final weighted estimation.}
		\label{fig:system}
	\end{figure}
	
	\subsection{Support for Standard and Flexible Analysis}
	
	For standard experiments, users provide minimal
	configuration---experiment IDs, variant lists, and
	references to pre-registered metric tables---and
	assignment queries are automatically constructed.
	For complex use cases, MEA supports custom metrics
	defined on the fly, arbitrary metric tables, and
	user-provided assignment queries for non-standard
	setups (e.g., complex bucketing, holdouts).
	A lightweight local mode accepts Pandas DataFrames
	directly for smaller-scale or prototyping use cases.

	\section{Deployment and Lessons Learned}

	Hundreds of MEA analyses have been run across more
	than eight product verticals.
	We audited approximately 50 documented analyses;
	the statistics below are drawn from that audit.

	Adoption started in the Premium (Subscriptions)
	business with a limited set of high-priority metrics.
	This phased rollout strategy proved essential for
	building user trust and iterating on the interface
	before broader expansion.
	An intake process where the data science team reviewed
	use case feasibility helped ensure MEA was applicable
	and users had appropriate expectations.

	Table~\ref{tab:usage} summarizes the primary usage
	categories observed across audited analyses.
	Categories overlap---a single analysis can serve
	multiple purposes (e.g., debugging that also informs
	a launch decision).

	\begin{table}[t]
	\centering
	\caption{Usage categories across audited MEA analyses.}
	\label{tab:usage}
	\small
	\begin{tabular}{lr}
	\toprule
	\textbf{Category} & \textbf{\%} \\
	\midrule
	Optimal combination identification & ${\sim}35\%$ \\
	Scenario-based conditional analysis & ${\sim}30\%$ \\
	Debugging / diagnostics & ${\sim}35\%$ \\
	Cross-team experiments & ${\sim}33\%$ \\
	Changed or informed decision & $80$--$85\%$ \\
	Independence flag triggered & ${\sim}4\%$ \\
	\bottomrule
	\end{tabular}
	\end{table}

	Debugging proved surprisingly valuable: even when
	MEA's independence assumptions are not fully satisfied,
	the overlap visualizations, triggering distributions,
	and metric breakdowns by variant combination help
	identify configuration errors, non-orthogonal
	assignments, and unexpected metric movements---making
	the tool useful for root-cause analysis beyond its
	primary statistical purpose.

	\subsection{Representative Use Cases}
	Beyond the detailed example in Section~3, the
	following vignettes illustrate the range of
	production applications:
	\begin{itemize}[nosep]
	\item \textbf{Launch unblocking.}
	  A monetization experiment could not launch while
	  a concurrent experiment was still running.
	  MEA showed the effect was robust to all
	  concurrent variants, unblocking launch early.
	\item \textbf{ML attribution.}
	  Three concurrent ML models had near-total user
	  overlap.
	  The standard platform attributed equal credit;
	  MEA identified one model as the primary
	  driver and revealed a $2\times$ attribution
	  difference---the team acknowledged
	  their standard reports were
	  ``biased by interactions.''
	\item \textbf{Debugging.}
	  Two experiments showed negative metric
	  impact simultaneously.
	  MEA isolated a platform migration as the
	  primary cause, avoiding a costly shutdown of
	  both experiments.
	\item \textbf{Cross-team analysis.}
	  Two experiments from different teams
	  shared the same page surface.
	  MEA showed one had no incremental
	  effect---the other was the dominant
	  driver across key metrics.
	\item \textbf{SSRM diagnosis.}
	  A concurrent change appeared to cause sample
	  ratio mismatch in another experiment.
	  MEA revealed treatment-induced triggering as
	  the mechanism, not a data error---the team
	  proceeded without pausing.
	\item \textbf{Migration safety.}
	  A platform-wide migration touched
	  surfaces shared by dozens of experiments.
	  MEA confirmed no material interaction,
	  clearing both the migration and all
	  experiments to ramp independently.
	\end{itemize}

	Workshops explaining MEA's methodology and
	interpretation guidelines improved adoption quality.
	Maintaining visual consistency with existing platform
	reports helped users trust and consume MEA outputs.

	\section{Conclusion and Future Directions}
	
	MEA enables reliable post-hoc analysis of
	concurrent experiments, delivering consistent individual,
	combinatorial, and conditional insights at scale.

	By addressing experimental interference from overlaps,
	MEA complements existing velocity tools---variance
	reduction for statistical efficiency and sequential
	testing for early stopping---by removing the need for
	traffic fragmentation or temporal sequencing without
	sacrificing statistical validity.
	
	While this work focuses on online experimentation in
	technology platforms, the methodology applies broadly
	to any domain where multiple interventions are tested
	concurrently on overlapping populations, including
	clinical trials \cite{green2002clinical}, public health,
	and manufacturing.
	
	Future work could explore several extensions. First,
	extending the framework to observational settings
	where clean randomization is unavailable would
	significantly broaden applicability. By integrating
	techniques from causal inference for observational
	data---such as propensity score methods, instrumental
	variables, or difference-in-differences designs
	\cite{hernan2020causal, imbens2015causal}---MEA's
	partitioning and weighted estimation approach could
	be adapted to handle confounding and selection bias
	in non-experimental contexts.

	Second, Appendix~\ref{sec:mea-sample-size} shows that,
	under mild assumptions, MEA's variance is at most that
	of a coordinated factorial design, with exponential
	gains as $k$ grows;
	extending this analysis to incorporate CUPED-style
	variance reduction \cite{deng2013improving} is a
	natural next step.

	Third, while automated assumption diagnostics are
	already deployed in every MEA report
	(Appendix~\ref{sec:assumption-checking}),
	a natural extension is automated interaction
	screening across all overlapping experiment pairs
	to surface unexpected interactions proactively.

	\paragraph*{Acknowledgments.}
		The methodology and implementation were developed
		by Reza Hosseini. The author thanks Saad Eddin
		Al Orjany, Yiwei Liu, Jilei Yang, and Pasha
		Khosravi for code and draft reviews, and Aman
		Grover, Prashanthi Padmanabhan, Aarthi Jayaram,
		Shahriar Pezeshgi, Ruonan Hao, Ming Wu, and
		Hector Duenas for helping drive adoption across
		teams.

	\appendix

\section{Assumption Checking}
\label{sec:assumption-checking}

MEA's correctness relies on a causal
assumption---\emph{Arm-Trigger Invariance} ---
which we state here, connect to a causal graph,
show what it guarantees, and describe the automated
test used in every production MEA report.

\subsection{Arm-Trigger Invariance}

Let $A_i$ denote the arm assignment
in experiment~$i$ (a latent label set before the
user triggers),
$S_i = \mathbf{1}\{\text{user triggers }E_i\}$
the trigger state,
$V_i = A_i$ when $S_i = 1$ and
$V_i = \texttt{nan}$ otherwise (the observed
variant), $U$ unobserved user characteristics,
and $Y$ the outcome.

\begin{assumption}[Arm-Trigger Invariance]
\label{asm:trigger-invariance}
$P(S \mid \mathrm{do}(A = a)) = P(S)$
for all arm combinations
$a \in \mathcal{A}_1 \times \cdots
\times \mathcal{A}_k$.
\end{assumption}

In words: fixing any particular arm combination
(via Pearl's $\mathrm{do}$-operator~\cite{pearl2009causality})
does not change the distribution of trigger states.
\textbf{Note:} trigger states $S_i$ and $S_j$ may be
freely correlated---and typically are, through
shared user characteristics $U$---the assumption
only requires that \emph{arm assignments} do not
shift trigger probabilities.

\subsection{Causal Graph Interpretation}

Figure~\ref{fig:causal-graph} shows the causal
graph for $k = 2$.
The assumption is equivalent to the \emph{absence}
of the dashed red edges $A_j \to S_i$ ($j \neq i$)
--- no experiment's arm assignment shifts another
experiment's trigger rate.

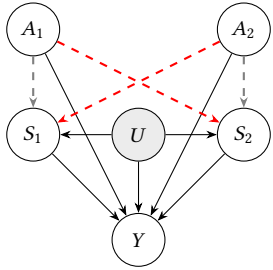
\begin{figure}[h]
\centering
\begin{tikzpicture}[
  >={Stealth[length=4pt]},
  every node/.style={
    font=\small,
    circle,
    draw,
    minimum size=20pt,
    inner sep=1pt},
  edge/.style={->},
  forbidden/.style={->,dashed,red,thick},
  selfabsent/.style={->,dashed,gray,thick}
]
  \node (A1) at (-1.4, 2.0) {$A_1$};
  \node (A2) at ( 1.4, 2.0) {$A_2$};
  \node (S1) at (-1.4, 0.6) {$S_1$};
  \node (S2) at ( 1.4, 0.6) {$S_2$};
  \node[fill=gray!15] (U) at (0, 0.6) {$U$};
  \node (Y) at (0, -0.8) {$Y$};
  \draw[edge] (U) -- (S1);
  \draw[edge] (U) -- (S2);
  \draw[edge] (U) -- (Y);
  \draw[edge] (A1) -- (Y);
  \draw[edge] (A2) -- (Y);
  \draw[edge] (S1) -- (Y);
  \draw[edge] (S2) -- (Y);
  \draw[selfabsent] (A1) -- (S1);
  \draw[selfabsent] (A2) -- (S2);
  \draw[forbidden] (A1) -- (S2);
  \draw[forbidden] (A2) -- (S1);
\end{tikzpicture}
\caption{Causal graph for $k{=}2$.
  $A_i$ = arm assignment (randomized);
  $S_i$ = trigger state;
  $U$ = unobserved user characteristics
  (shaded); $Y$ = outcome.
  Dashed edges are absent under
  Assumption~\ref{asm:trigger-invariance}:
  \textcolor{gray}{gray} $A_i \to S_i$ holds by
  experiment design (arm is assigned before
  triggering) and needs no test;
  \textcolor{red}{red} $A_j \to S_i$ is the
  non-trivial condition tested in production.}
\label{fig:causal-graph}
\end{figure}

Three sufficient conditions ensure the assumption holds:
\begin{enumerate}[nosep]
\item[(C1)] \emph{Randomization:}
  $A \perp\!\!\!\perp U$---arm assignments are
  independent of user
  characteristics~\cite{rubin1974estimating}.
\item[(C2)] \emph{No cross-trigger contamination:}
  $A_j \not\to S_i$ for all $j \neq i$ ---
  no experiment's arm shifts another experiment's
  trigger rate.
\item[(C3)] \emph{No within-trigger contamination:}
  $A_i \not\to S_i$ for all $i$ ---
  arm assignment does not affect the user's own
  trigger probability.
\end{enumerate}
(C1) is guaranteed by
randomization~\cite{rubin1974estimating}.
(C3) holds by within-experiment design: $A_i$ is
a latent label assigned before the user triggers,
and the trigger condition is defined on
pre-treatment behavior (e.g., user visits page~$X$;
treatment is what they see on page~$X$).
(C2) is where violations occur in practice and
is the target of the statistical test below.

Under (C1)--(C3), $S_i$ has no $A$-parent in the
causal graph, so $S = f(U)$ and
$S \perp\!\!\!\perp A$.

\subsection{What the Assumption Guarantees}

Arm-Trigger Invariance, combined with standard
platform design, guarantees the two conditions
needed for consistent weighted estimation:
\begin{enumerate}[nosep]
\item \emph{Stable stratum weights.}
  $P(S = s)$ does not shift when a particular arm
  combination is launched, so the weights $w_r$
  target the correct population mixture.
  This follows directly from
  Assumption~\ref{asm:trigger-invariance}.
\item \emph{Unbiased within-stratum estimation.}
  Arm assignments are independent within each
  stratum, so cell means are unconfounded.
  This follows from
  Assumption~\ref{asm:trigger-invariance}
  combined with independent hashing:
\end{enumerate}

\begin{proposition}
Under Assumption~\ref{asm:trigger-invariance}
and independent
hashing~\cite{tang2010overlapping},
$A_i \perp\!\!\!\perp A_j \mid S = s$
for all $i \neq j$ with $s_i = s_j = 1$.
\end{proposition}
\emph{Proof.}
Independent hashing gives
$A_i \perp\!\!\!\perp A_j$
unconditionally~\cite{tang2010overlapping}.
Assumption~\ref{asm:trigger-invariance}
gives $S \perp\!\!\!\perp A$.
Hence $P(A_i, A_j \mid S{=}s)
= P(A_i)\,P(A_j)
= P(A_i \mid S{=}s)\,P(A_j \mid S{=}s)$.
\hfill$\square$

\subsection{Automated Diagnostic}

The visual and statistical tests below verify
both conditions simultaneously: the
\texttt{nan} columns test trigger
contamination~(C2), while the
non-\texttt{nan} columns test assignment
independence within $\mathcal{R}_{11}$
--- catching platform-level randomization
bugs even when~(C2) holds.

\subsection{$k$ Joint Tests}

For $k$ experiments, the diagnostic runs $k$
chi-squared homogeneity tests---one per
\emph{source} experiment~$j$.
Each test fixes experiment~$j$ and asks:
does $j$'s arm assignment shift the joint
distribution of all other experiments' variants?

Concretely, for source~$j$, restrict to users
with $V_j \neq \texttt{nan}$ (triggered in~$j$)
and build a contingency table whose rows are
$j$'s arm values and whose columns are all
joint combinations of $(V_i)_{i \neq j}$,
\emph{including \texttt{nan}} in each
dimension.
The \texttt{nan}-containing columns test
condition~(C2)---whether $j$'s arm shifts
other experiments' trigger rates---while
the non-\texttt{nan} columns simultaneously
verify within-stratum assignment
independence (the Proposition above).
A single test per source covers both.
For $k = 2$ this produces 2~tests; for
$k = 3$, 3~tests with Bonferroni correction
$\alpha/k$.

MEA reports include one grouped bar chart
per source (Figure~\ref{fig:assumption-check}).
Under both conditions, all bars within each
group should be approximately equal in height.
Unequal \texttt{nan}-column heights signal
trigger contamination---the most common
violation in production.

\subsection{Statistical Test}

Each source test uses the chi-squared statistic
on the contingency table described above.
The only structurally absent cell is the
all-\texttt{nan} corner (members who triggered
no experiment are not in the dataset).
Expected counts are computed via block
decomposition, preventing the absent corner
from distorting expected values.

At experiment scale, the $p$-value alone is insufficient:
with $N$ in the tens of thousands or more, even a trivially
small imbalance will produce $p \ll 0.05$.
We therefore report Cram\'{e}r's $V$ alongside the $p$-value:
\[
  V = \sqrt{\frac{\chi^2}{N \cdot \min(k_1,\, k_2)}}
\]
where $N$ is the total count over the table (excluding the
structurally absent cell) and $\min(k_1, k_2)$ is the
effective degrees-of-freedom
scaling~\cite{agresti2002categorical}.
$V$ is bounded in $[0, 1]$ and does not grow with $N$,
making it directly interpretable as an effect size:
$V \approx 0$ means independence; $V \approx 1$ means
perfect dependence.
As a rough guide: $V < 0.01$ is negligible, $V \approx 0.1$
is small, $V \approx 0.3$ is moderate.

Following the convention of sample ratio mismatch (SRM)
testing in online experimentation~\cite{kohavi2020trustworthy,
fabijan2019diagnosing}, a flag is raised when \emph{both}
$p < \alpha/k$ (Bonferroni-corrected at
$\alpha = 0.05$) \emph{and} $V > 0.01$.
The dual threshold prevents two failure modes:
the $p$-value catches real signal even at moderate $N$,
while the $V$ threshold filters out statistically
significant but practically negligible deviations.
Flags prompt investigation, not automatic invalidation.

\subsection{Production Examples}

\textbf{Example 1: Triggering violation (FLAG).}
Two notification experiments (several hundred
thousand members triggering $E_1$).
Treatment in $E_1$ dramatically increases the rate
at which members trigger $E_2$ (96\% vs.\ 79\%
non-triggered; Figure~\ref{fig:assumption-check}),
while assignment balance within $\mathcal{R}_{11}$
is near-uniform.
$\chi^2 = 46{,}330$ ($p \approx 0$),
$V = 0.262$---FLAG.

\textbf{Example 2: Assumption holds (PASS).}
The content test ($E_1$: three variants) and page redesign
($E_2$: two variants) from the main body use case
(Section~3).
The \texttt{nan} rates are virtually identical across
all $E_1$ variants (${\approx}\,91\%$ for each),
and the conditional $E_2$-variant distribution within
$\mathcal{R}_{11}$ shows negligible variation.
$\chi^2 = 0.21$ ($p = 0.995$), $V = 0.013$---PASS.

\subsection{Production Experience}

Across documented production analyses, approximately
4\% triggered the independence flag.
Violations fall into two categories:
treatment-induced triggering (Example~1 above), where
the experimental design creates a structural dependency
between triggering events; and assignment errors, where
a misconfigured assignment boundary causes correlated
variant draws.
The 96\% pass rate suggests the check is appropriately
calibrated: sensitive enough to catch real issues while
not blocking routine analyses.

\section{Simulation Validation}
\label{sec:simulations}

We validate MEA through comprehensive simulations
with known ground truth effects. The simulation
setup involves two overlapping experiments with
the following characteristics:

\subsection{Simulation Setup}

We simulate a heterogeneous population with two
experiments: $E_1$ (control, v1, v2; 50\% trigger rate;
variant split 40/30/30) and $E_2$ (control, enabled;
40\% trigger rate; variant split 40/60).
Triggering is independent across experiments.
Univariate effects on metric1: $E_1$ v1 ($-2$),
v2 ($+5$); $E_2$ enabled ($-2$).
Interactions: (v1, enabled) $+15$;
(v2, control) $-2$.
We focus on estimating the launch impact of
(v1, enabled) vs.\ (control, control), whose true
value depends on the overlap structure.

\subsection{Deriving the True Expected Delta}

Under independent triggering, the region
probabilities and per-region effects of
(v1, enabled) vs.\ (control, control) are:
$\mathcal{R}_{11}$: $P = 0.20$, effect $= +11$;
$\mathcal{R}_{10}$: $P = 0.30$, effect $= -2$;
$\mathcal{R}_{01}$: $P = 0.20$, effect $= -2$;
$\mathcal{R}_{00}$: $P = 0.30$, effect $= 0$.
Conditioning on the impacted population
($P = 0.70$), the true expected delta is:
\begin{align}
	\mathbb{E}[\Delta \mid \text{impacted}]
	&= \frac{0.20 \times 11 + 0.30 \times (-2) + 0.20 \times (-2)}{0.70}
	= \frac{1.2}{0.70} \approx 1.71
\end{align}
This masks substantial heterogeneity: $+11$ in
$\mathcal{R}_{11}$ (29\% of impacted) vs.\ $-2$ in the
singly-triggered regions (71\%).

\subsection{Validation of Expected Delta Computation}

Figure~\ref{fig:sim_validation} confirms that
simulation-based estimates converge to $+1.71$
as sample size increases, validating the derivation.

\begin{figure}[t]
	\centering
	\includegraphics[width=0.95\columnwidth]{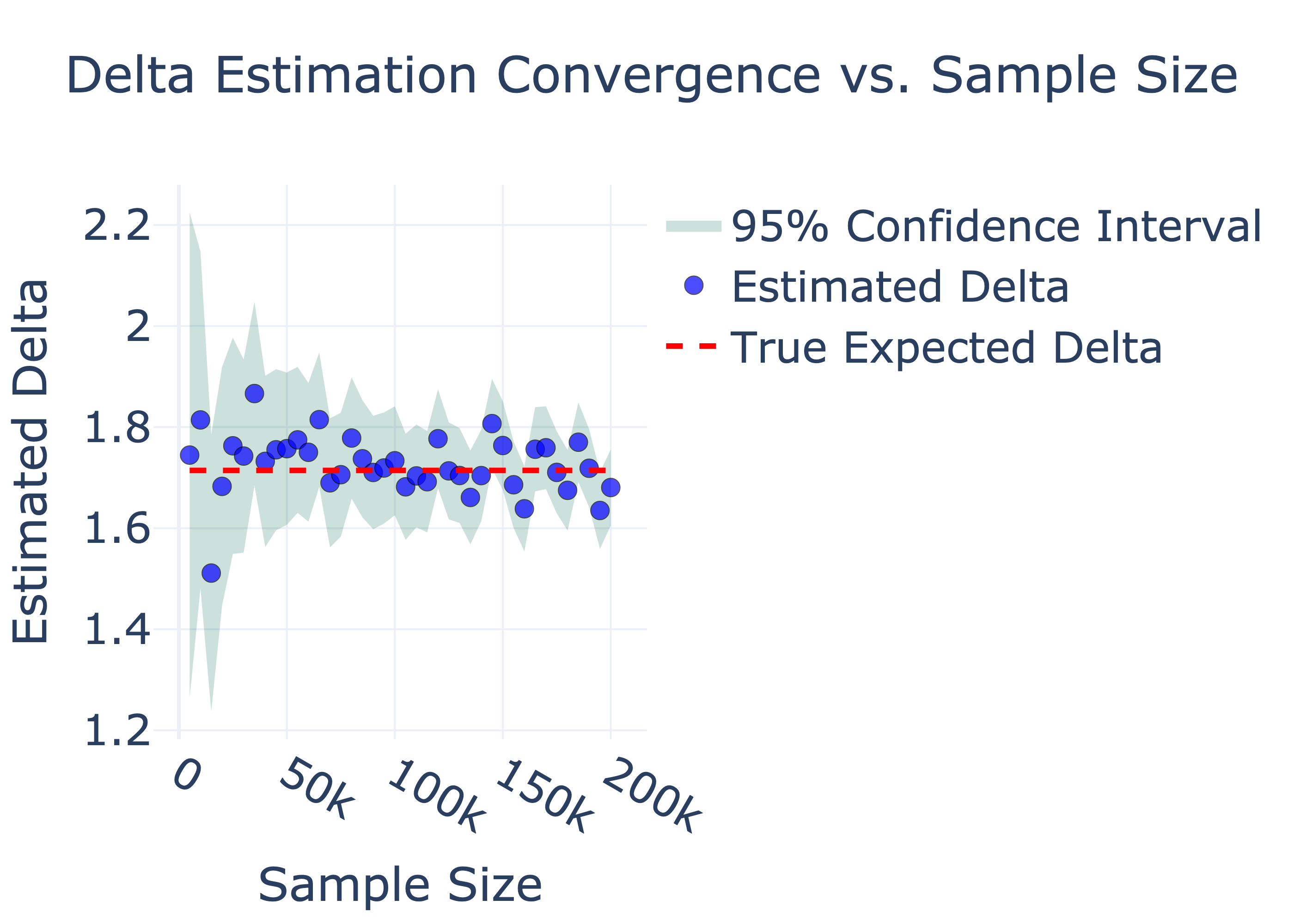}
	\caption{Convergence of simulation-based delta estimates
		to the mathematically expected delta ($+1.71$). The
		estimates converge to the expected value as sample
		size increases, validating the theoretical computation.}
	\label{fig:sim_validation}
\end{figure}

\subsection{Consistency and Coverage}

We conducted 1000 independent simulations with
randomized attributes and assignments.
Figure~\ref{fig:sim_histogram} shows the distribution
of MEA estimates centered on the true delta, with
one example 95\% CI overlaid.

\begin{figure}[t]
	\centering
	\includegraphics[width=0.99\columnwidth]{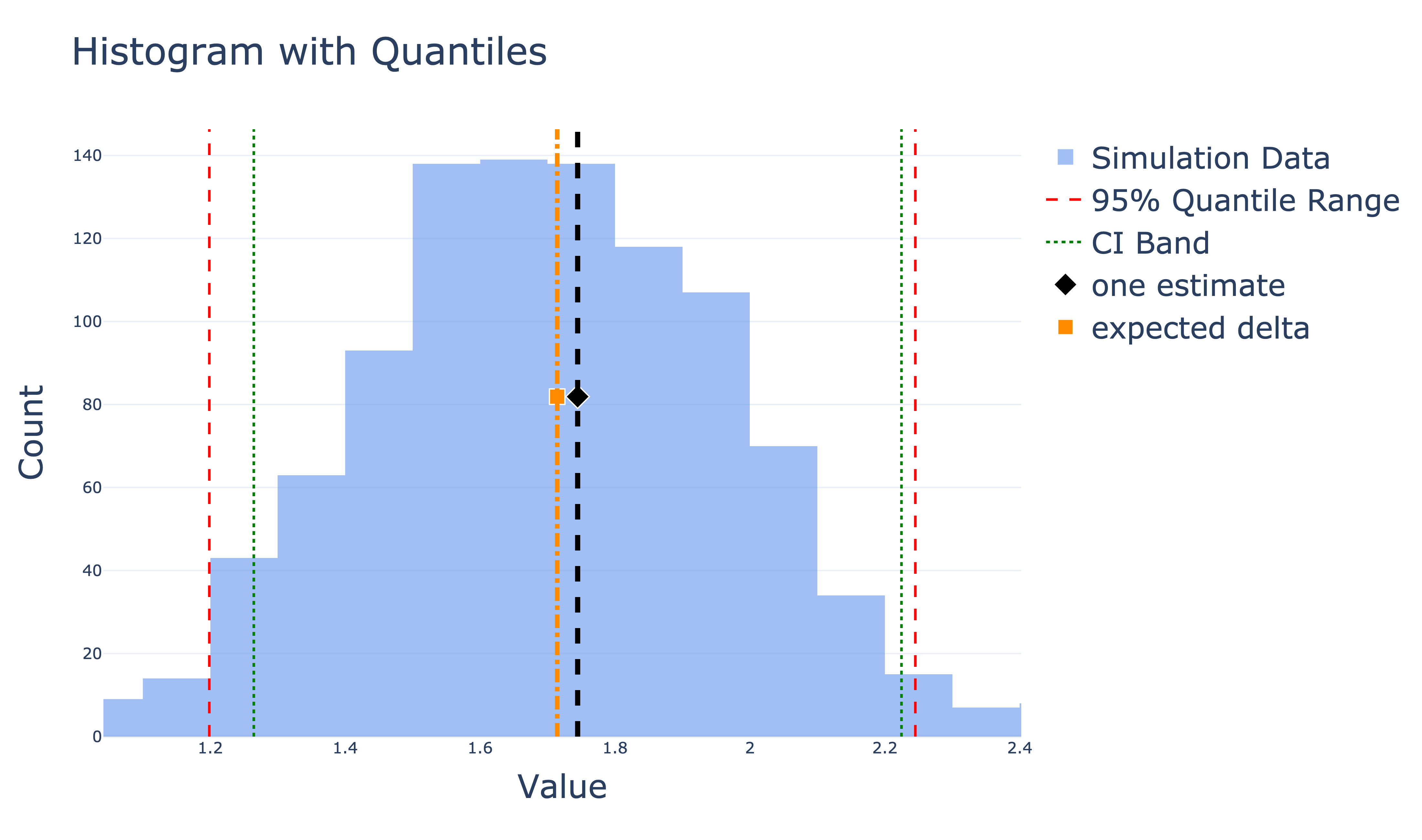}
	\caption{Distribution of 1000 MEA estimates with
		one example 95\% confidence interval (bands).
		The estimates are centered on the true expected
		delta (red dashed line at $+1.71$), confirming consistency.}
	\label{fig:sim_histogram}
\end{figure}

Across 1000 simulations, the 95\% confidence intervals
contained the true parameter 95.2\% of the time,
confirming correct coverage. The average CI length
closely matched the empirical 95\% quantile range
of the estimates, indicating well-calibrated
uncertainty quantification.

\subsection{Convergence with Sample Size}

Figure~\ref{fig:sim_convergence} shows MEA estimates
converging to the true parameter as sample size
increases from 5K to 200K, with CIs narrowing
appropriately.

\begin{figure}[t]
	\centering
	\includegraphics[width=0.99\columnwidth]{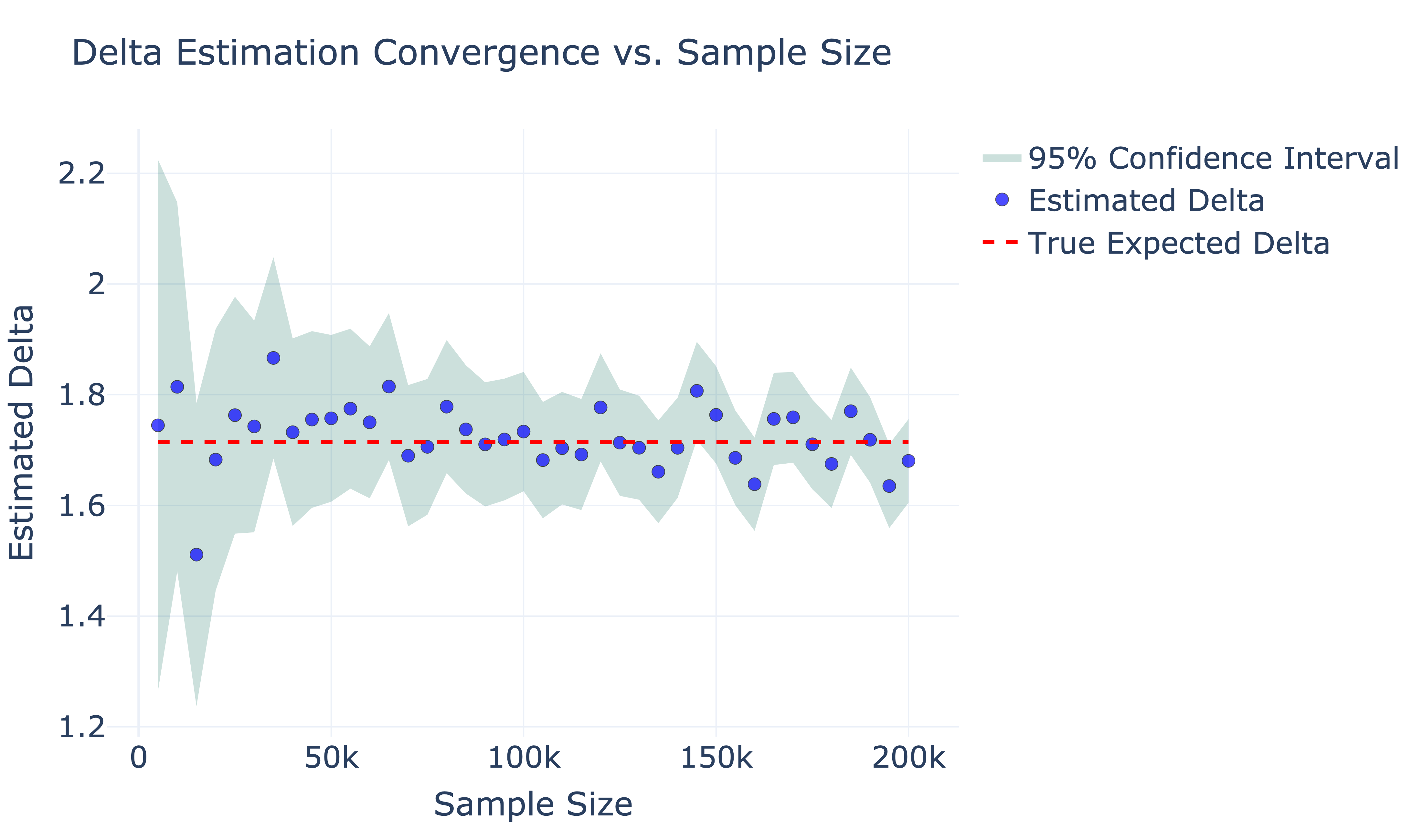}
	\caption{MEA estimate convergence as sample size
		increases from 5K to 200K units. Points show
		individual estimates, shaded region shows 95\%
		confidence intervals, and red dashed line
		indicates true expected delta ($+1.71$).}
	\label{fig:sim_convergence}
\end{figure}

Additional simulations---omitted for brevity---validate
MEA's consistency and coverage for ratio metrics
(e.g., click-through rate, conversion rate) using
both Delta method and bucketed jackknife variance
estimation, confirm convergence of jackknife
standard errors to their Delta method counterparts
as bucket count increases, and verify correct
coverage under dependent triggering structures
(e.g., near-total nesting).

\subsection{Comparison with Regression}

A natural alternative is interaction regression
$Y = \beta_0 + \beta_1 A_1 + \beta_2 A_2
+ \beta_{12} A_1 A_2 + \varepsilon$
on $\mathcal{R}_{11}$.
This is unbiased within $\mathcal{R}_{11}$ and
converges to that region's effect ($+11$), but the
true launch impact across all triggered users is
$+1.71$ (Table~\ref{tab:reg-sim}).
The $6\times$ overestimate arises because
regression ignores $\mathcal{R}_{10}$ and
$\mathcal{R}_{01}$, where the combination effect
is negative.
Including singly-triggered users by coding
non-triggered arms as $0$ contaminates the
reference cell with heterogeneous baselines,
introducing bias.

\begin{table}[t]
\centering
\caption{Regression vs.\ MEA: $(v1, \text{enabled})$
launch impact.}
\label{tab:reg-sim}
\small
\begin{tabular}{lcc}
\toprule
\textbf{Approach}
& \textbf{Data used}
& \textbf{Estimate} \\
\midrule
Regression on $\mathcal{R}_{11}$
  & $\mathcal{R}_{11}$ only
  & $+11.0$ \\
MEA
  & All triggered
  & $+1.71$ \\
\midrule
True launch impact
  & ---
  & $+1.71$ \\
\bottomrule
\end{tabular}
\end{table}

Additionally, regression does not naturally extend
to ratio metrics (e.g., CTR, conversion rate),
which require joint estimation of numerator and
denominator with appropriate variance correction.
MEA handles these via the Delta method and
bucketed jackknife
(Section~\ref{sec:ratio-metrics}).

\section{Univariate Analysis vs.\ MEA}
\label{sec:concurrent-vs-sequential}

A natural question is whether standard analysis approaches---either analyzing concurrent experiments independently
or running them sequentially in time---can match MEA's
joint analysis.
We construct a concrete scenario showing that both
alternatives recommend the \emph{worst} launch combination,
while MEA correctly identifies the optimum.

\subsection{Setup}

Two binary experiments (control/treatment each) run
concurrently on a shared population with
\emph{partial overlap}: each experiment triggers on
approximately~30\% of the population, producing
overlap regions
$\mathcal{R}_{11} \approx 9\%$,
$\mathcal{R}_{10} \approx 21\%$, and
$\mathcal{R}_{01} \approx 21\%$.

Ground-truth causal effects on the primary metric
(relative to all-control baseline):
\begin{itemize}[nosep]
  \item Experiment~1 treatment: $+3$
  \item Experiment~2 treatment: $+4$
  \item Interaction when both treatments active:
        $-10$ (strong negative synergy)
\end{itemize}

\noindent
Within the doubly-triggered region~$\mathcal{R}_{11}$, the
cell-level expected metric values are:
\begin{center}
\begin{tabular}{lcc}
\toprule
 & Expt 2 = control & Expt 2 = treatment \\
\midrule
Expt 1 = control   & $0$ (baseline)
  & $+4$ \textbf{(best)} \\
Expt 1 = treatment  & $+3$
  & $3+4-10 = -3$ \textbf{(worst)} \\
\bottomrule
\end{tabular}
\end{center}

Each treatment individually helps, but the combination
reverses the gains.  Table~\ref{tab:cell-effects-by-region}
shows the causal effect of each launch combination in each
triggering region.
The $-10$ interaction only activates in~$\mathcal{R}_{11}$
when both treatments are active.

\begin{table}[t]
\centering
\caption{Effect of each launch combination by triggering
  region (30\% trigger rates).}
\label{tab:cell-effects-by-region}
\begin{tabular}{lccc}
\toprule
\textbf{Launch combo}
  & \makecell{$\mathcal{R}_{10}$\\($\approx 21\%$)}
  & \makecell{$\mathcal{R}_{01}$\\($\approx 21\%$)}
  & \makecell{$\mathcal{R}_{11}$\\($\approx 9\%$)} \\
\midrule
$(c_1, t_2)$ & --- & $+4$ & $+4$ \\
$(t_1, c_2)$ & $+3$ & --- & $+3$ \\
$(t_1, t_2)$ & $+3$ & $+4$ & $-3$ \\
\bottomrule
\end{tabular}
\end{table}

\subsection{Three Approaches Compared}

\paragraph{Approach 1: Concurrent univariate.}
Both experiments run simultaneously, each analyzed
in isolation.
Each experiment's univariate treatment effect pools
across all regions where it triggers.
For Experiment~1, the triggered population is
$\mathcal{R}_{10} \cup \mathcal{R}_{11}$.
In~$\mathcal{R}_{10}$ ($\approx 21\%$),
the effect is a clean~$+3$.
In~$\mathcal{R}_{11}$ ($\approx 9\%$), Experiment~2 is
active with a 50/50 split, so the average
Experiment~1 effect is
$3 + 0.5 \times (-10) = -2$.
Pooling:
\[
  \widehat{\tau}_1^{\,\mathrm{univar}}
  \;\approx\;
  \frac{21}{30} \times 3
  \;+\;
  \frac{9}{30} \times (-2)
  \;=\;
  +1.5
  \;>\;0.
\]
An analogous calculation gives
$\widehat{\tau}_2^{\,\mathrm{univar}} \approx +2.5 > 0$.
Both positive $\to$ decision: launch $(t_1, t_2)$.

\paragraph{Approach 2: Sequential in time.}
Run Experiment~1 first, ship its winner ($t_1$),
then run Experiment~2 with $t_1$ deployed.
\begin{enumerate}[nosep]
  \item Experiment~1 analyzed alone
    (univariate effect $= +1.40 > 0$).
    Decision: ship $t_1$.
  \item $t_1$ is now permanent for all Experiment~1
    users.  $\mathcal{R}_{10}$ is no longer experimental.
  \item Experiment~2 runs on
    $\mathcal{R}_{01} \cup \mathcal{R}_{11}$.
    In $\mathcal{R}_{01}$ ($\approx 70\%$ of Expt~2's
    population): $t_1$ has no effect, clean $+4$.
    In $\mathcal{R}_{11}$ ($\approx 30\%$): $t_1$ already
    deployed, conditional effect $\approx -6$.
  \item Pooled: $0.70 \times (+4) + 0.30 \times (-6)
    = +0.98 > 0$.
  \item Decision: ship $t_2$.
    Final outcome: $(t_1, t_2)$---the worst combination.
\end{enumerate}

\paragraph{Approach 3: MEA joint.}
Both experiments run simultaneously, analyzed jointly
across all triggering regions.
MEA evaluates all launch combinations and identifies
$(c_1, t_2)$ as optimal.

\subsection{Simulation Results}

We validate with $N = 100{,}000$ members using
independent hashing for assignment, 30\% trigger rates,
and noise $\sigma = 5.0$.
Table~\ref{tab:three-way-comparison} compares all three
approaches.

\begin{table}[t]
\centering
\caption{Three-way comparison ($N = 100{,}000$,
  30\% trigger rates).}
\label{tab:three-way-comparison}
\begin{tabular}{llc}
\toprule
\textbf{Approach} & \textbf{Decision} & \textbf{Effect} \\
\midrule
Concurrent univariate
  & $(t_1, t_2)$
  & $+2.25$ \\
Sequential (Expt 1 first)
  & $(t_1, t_2)$
  & $+2.25$ \\
\textbf{MEA joint}
  & $\mathbf{(c_1, t_2)}$
  & $\mathbf{+3.87}$ \\
\bottomrule
\end{tabular}
\end{table}

\noindent
MEA's combination estimates, validated against
analytically computed ground truth:
\begin{center}
\begin{tabular}{lcc}
\toprule
\textbf{Combination}
  & \textbf{Ground truth}
  & \textbf{MEA estimate} \\
\midrule
$(c_1, t_2)$
  & $+4.00$ & $+3.87$ \quad \textbf{(best)} \\
$(t_1, c_2)$
  & $+3.00$ & $+2.85$ \\
$(t_1, t_2)$
  & $+2.35$ & $+2.25$ \quad \textbf{(worst)} \\
\bottomrule
\end{tabular}
\end{center}

\subsection{Discussion}

The failure of both alternatives stems from
\emph{dilution}: univariate analysis averages over
singly-triggered regions where the interaction does not
apply, masking a harmful combination effect in
$\mathcal{R}_{11}$.
This arises whenever experiments partially overlap and
have a strong negative interaction---conditions that
are common in practice.
Only MEA's post-stratified joint analysis, which
separately estimates effects in each triggering region,
correctly identifies the launch-combination optimum.


\section{Power Analysis}
\label{sec:mea-sample-size}

\paragraph{Notation.}
Consider $k$ concurrently running experiments on an
eligible population of size $N$.
Experiment $j$ has $\ell_j$ variants (including control).
A \emph{triggering stratum} is a binary vector
$s = (s_1, \ldots, s_k) \in \{0,1\}^k$,
where $s_j = 1$ means the member triggered
experiment $j$.
Let $d(s) = \sum_j s_j$ denote the \emph{degree}
of $s$ (number of active experiments).
We exclude the zero stratum $s = 0$.

\begin{itemize}[nosep]
\item $N_s$ = number of members in stratum $s$;
  \; $N_+ = \sum_{s \neq 0} N_s$ = total triggered
  population.
\item $w_s = N_s / N_+$ = weight of stratum $s$
  among triggered members; $\sum_{s \neq 0} w_s = 1$.
\item $\sigma^2$ = within-cell outcome variance,
  assumed constant across treatment assignments
  (homoscedasticity):
  $\mathrm{Var}(Y \mid A{=}a) = \sigma^2$
  for all $a$.
\item $\sigma_s^2$ = within-cell outcome variance
  in stratum $s$:
  $\mathrm{Var}(Y \mid A{=}a,\, S{=}s) = \sigma_s^2$.
  Define
  $\tilde\sigma^2 = \sum_s w_s \sigma_s^2$
  (weighted average within-cell variance).
  By the law of total variance,
  $\tilde\sigma^2 \leq \sigma^2$
  (conditioning on stratum can only reduce
  within-cell variance).
\end{itemize}

A \emph{coordinated factorial design} randomizes
all $N_+$ users jointly across all $k$ experiments,
creating $\prod_{j=1}^{k} \ell_j$ cells with
$N_+ / \prod_{j=1}^{k} \ell_j$ users per cell.
Under homoscedasticity, the variance of any
combination-effect estimate is
$\mathrm{Var}_{\mathrm{fac}}
= 2\sigma^2 \prod_{j=1}^{k} \ell_j \,/\, N_+$.

\paragraph{MEA variance.}
From Section~2.3, the MEA estimator is
$\hat\tau = \sum_{s \neq 0} w_s \hat\delta_s$.
Under balanced randomization, stratum $s$ has
$\prod_{j:\,s_j=1} \ell_j$ cells, so
the within-stratum variance is
$2\sigma_s^2 \prod_{j:\,s_j=1} \ell_j / N_s$.
Substituting $w_s = N_s/N_+$:
\[
\mathrm{Var}_{\mathrm{MEA}}
= \frac{2}{N_+}
  \sum_{s \neq 0} w_s\,\sigma_s^2
  \prod_{j:\,s_j=1} \ell_j.
\]
This treats weights as fixed; estimated weights add
$O(\sigma^2 \prod \ell_j / N^2)$, negligible provided
$\prod \ell_j \ll N$ (stable weights).

\paragraph{MEA is always at most as expensive as factorial.}
Since $\prod_{j:\,s_j=1} \ell_j \leq \prod_{j=1}^{k} \ell_j$
for every $s$ (partial overlap has fewer cells)
and $\tilde\sigma^2 \leq \sigma^2$ (established above):
\[
\mathrm{Var}_{\mathrm{MEA}}
\;\leq\;
\frac{2\,\tilde\sigma^2
  \prod_{j=1}^{k} \ell_j}{N_+}
\;\leq\;
\frac{2\,\sigma^2 \prod_{j=1}^{k} \ell_j}{N_+}
= \mathrm{Var}_{\mathrm{fac}}.
\]

\paragraph{How large is the difference?}
Assume $\sigma_s^2 = \sigma^2$, $\ell_j = \ell$, and
independent triggering at equal rates $r$.
\begin{align*}
\mathrm{Var}_{\mathrm{MEA}}
  &= \frac{2}{N_+}
     \sum_{s \neq 0} w_s\,\sigma_s^2
     \prod_{j:\,s_j=1} \ell_j
  \\[4pt]
  &= \frac{2\sigma^2}{N_+}
     \sum_{s \neq 0} w_s\, \ell^{d(s)}
  && (\sigma_s^2 = \sigma^2,\;
      \textstyle\prod_{j:\,s_j=1} \ell = \ell^{d(s)})
  \\[4pt]
  &= \frac{2\sigma^2}{N_+}
     \cdot
     \frac{\sum_{s \neq 0}
       (\ell r)^{d(s)}(1{-}r)^{k-d(s)}}
          {1 - (1{-}r)^k}
  && (w_s = \tfrac{r^{d(s)}(1{-}r)^{k-d(s)}}
            {1-(1{-}r)^k})
  \\[4pt]
  &= \frac{2\sigma^2}{N_+}
     \cdot
     \frac{(1{+}(\ell{-}1)r)^k - (1{-}r)^k}
          {1 - (1{-}r)^k}.
  && (\text{product-of-sums})
\end{align*}

The ratio:
\begin{equation*}
\frac{\mathrm{Var}_{\mathrm{MEA}}}
     {\mathrm{Var}_{\mathrm{fac}}}
= \frac{(1{+}(\ell{-}1)r)^k - (1{-}r)^k}
       {\ell^k(1 - (1{-}r)^k)}.
\end{equation*}

For binary experiments ($\ell = 2$, $r = 1/2$),
the ratio simplifies to
$(3^k - 1)\,/\,{2^k(2^k - 1)}
\approx (3/4)^k$:
e.g.\ 0.67 at $k{=}2$,
0.46 at $k{=}3$,
0.33 at $k{=}4$,
0.24 at $k{=}5$.
The savings grow exponentially with $k$,
and are larger at lower trigger rates.

\end{document}